\title{\textbf{Two-sample Testing on Latent Distance Graphs With Unknown Link Functions}\\[0.3in]}
\date{\today}

\documentclass[12pt]{article}
\usepackage{setspace}
\usepackage[toc,page]{appendix} 
\usepackage{subcaption}
\usepackage{amsmath}
\usepackage{amsthm}
\usepackage{amsfonts}
\usepackage{graphicx}
\usepackage{booktabs}
\usepackage{caption}
\usepackage{float}
\usepackage{lineno,hyperref}
\usepackage{xcolor}
\hypersetup{
    colorlinks,
    linkcolor={red!50!black},
    citecolor={magenta!80!black},
    urlcolor={blue!80!black}
}

\usepackage{multirow}
\usepackage[outdir=./]{epstopdf}
\usepackage[export]{adjustbox}
\usepackage[margin=1.2in]{geometry}
\usepackage{siunitx}
\usepackage{array}
\usepackage{color, colortbl}
\usepackage[section]{placeins}
\usepackage{newunicodechar}
\usepackage[utf8]{inputenc}
\usepackage{textcomp}
\usepackage{xcolor}
\usepackage[ruled,vlined]{algorithm2e}
\usepackage{setspace}
\usepackage{pdfpages}
\usepackage{enumitem}
\usepackage{breakcites}
\usepackage{mathtools}

\DeclarePairedDelimiter{\ceil}{\lceil}{\rceil}
\newcolumntype{L}[1]{>{\raggedright\arraybackslash}p{#1}}

\newtheorem{lemma}{Lemma}

\newtheorem{theorem}{Theorem}

\newtheorem{definition}{Definition}
\newtheorem{corollary}{Corollary}
\newtheorem{assumption}{Assumption}

\newtheorem{remark}{Remark}
\numberwithin{theorem}{section}
\numberwithin{equation}{section}

\usepackage{authblk}
\author[1]{Yiran Wang}
\author[2]{Minh Tang}
\author[3]{Soumendra Nath Lahiri}
\affil[1,2]{Department of Statistics, North Carolina State University}
\affil[3]{Department of Mathematics and Statistics, Washington University in St. Louis}

\title{Two-sample hypothesis testing for latent distance graphs with unknown link functions}

\begin{document}

\maketitle
\begin{abstract} We propose a valid and consistent test for the
hypothesis that two latent distance random graphs on the same
vertex set have the same generating latent positions, up to some
unidentifiable similarity transformations. Our test statistic is based
on first estimating the edge probabilities matrices by
truncating the singular value decompositions of the averaged adjacency
matrices in each population
and then computing a Spearman rank correlation coefficient between these
estimates. Experimental results on simulated data indicate that the
test procedure has power even when there is only one sample from each
population, provided that the number of vertices is not
too small. Application on a dataset of neural connectome graphs showed
that we can distinguish between scans from different age groups while
application on a dataset of epileptogenic recordings showed that we
can discriminate between seizure and non-seizure events. 
\end{abstract}

{\bf Keywords:} 
Graph inference; Latent distance graphs model; Two-sample hypothesis testing.

\newpage
\section{Introduction}

In recent years, the increasing popularity of
network data in diverse fields has spurred significant developments in
many theoretical and applied research related to random graph models
and their statistical inference
\cite{erd6s1960evolution,hoff2002latent,handcock2007model,wasserman1996logit,holland1983stochastic, airoldi2008mixed, karrer2011stochastic}. A
significant amount of literature on statistical inference for random
graphs has focused on estimation \cite{chatterjee2015matrix, xu2017rates, Olhede14722} and community detection; see
\cite{abbe2017community} and the references therein for a survey of
recent progresses on community detection.

In contrasts, the problem of graph comparisons or two-sample hypothesis
testing on random graphs has not been as well studied in
statistics. Graph comparisons are widely used in neuroscience, with
two prominent approaches. One approach advocates comparing the edges
directly  as a collection of paired
$t$-tests \cite{zalesky2010network} while the other proposes to
compare graph-theoretic measures such as the clustering coefficient,
path length, and their ratio
\cite{rubinov2010complex,he2008structural,humphries2008network}. These
two approaches, while useful, have their own disadvantages. In particular, the
pairwise comparison of edges views a network on $n$ vertices as a
collection of $n(n-1)/2$ edges and ignores any underlying network
structure or topology. The use of graph-theoretic measures, meanwhile, 
implicitly assumes that a graph can be reasonably summarized by a few graph invariants. It is, however, not {\em a priori}
clear which invariants are appropriate and oftentimes a graph
invariant is chosen due to its computational cost, e.g., number of triangles versus general cliques. Finally,
neither of these approaches consider the statistical implications in
term of validity and consistency of the test procedures. 

In statistics literature, several methods have been investigated
recently. \cite{ginestet2017hypothesis} derived a
central limit theorem for the sample Fr\'{e}chet mean of combinatorial
graph Laplacians and built a Wald-type two-sample test
statistic. \cite{ghoshdastidar2017two} proposed to compare the
underlying graph-generating distributions with a test statistic based
on differences of the estimated edge-probability matrices with respect
to the spectral or Frobenius norms. \cite{ghoshdastidar2018practical}
further developed a test statistic via extreme eigenvalues of a
scaled and centralized matrix as motivated by the Tracy-Widom
law. The aforementioned literature generally does not assume any specific generative model for the observed graphs. For example,
\cite{ginestet2017hypothesis} only assumes that the vertices are
aligned while \cite{ghoshdastidar2017two} and \cite{ghoshdastidar2018practical}
only require conditional independence of edges.  There is then an
inherent tradeoff between generality of the generative model and
specificity of the theoretical results; indeed, for these models,
a graph on $n$
vertices could require $n(n-1)/2$ parameters for the pairwise edge
probabilities. To address the potential need for estimating these
parameters, \cite{ginestet2017hypothesis} assume that the number of graphs is
reasonably large compared to the number of vertices, while the
necessary conditions for consistency of the test statistics in
\cite{ghoshdastidar2017two} and \cite{ghoshdastidar2018practical} are quite
complex.

Various popular generative graph models such as the stochastic block
model, latent space model and their variants, have also been actively
studied in the context of two-sample testing problems.
Stochastic block model graphs, first introduced by
\cite{holland1983stochastic}, assume that vertices are partitioned
into several unobserved blocks and the probability of connection is a
function of block membership. Under this model, \cite{li2018two} studied the problem of
testing the differences of block memberships and constructed test statistic via singular subspace
distance. The notion of hidden communities in stochastic blockmodel
graphs can be generalized to yield latent space model graphs or graphons
\cite{hoff2002latent,bollobas_riordan,lovasz2012}.  In the latent
space model each vertex $v_i$ is associated with a latent position
$x_i \in \mathbb{R}^{d}$ and, conditioned on the latent positions of
these vertices, the edges are independent Bernoulli random variables with
mean $p_{ij}=f(x_i,x_j)$ where $f$ is a symmetric link
function. 
A special case of a latent space model is the notion of a
generalized random dot product graph
\cite{young2007random,delanchy_grdpg} wherein $f(x_i, x_j) = \langle
x_i, x_j \rangle$ for some inner product or bilinear form $\langle
\cdot, \cdot \rangle$. Random dot product graphs include, as a
special case, stochastic blockmodel graphs and its degree-corrected
and mixed-membership variants, and furthermore, any latent position
graphs can be represented as a random dot product graph with a fixed
dimension $d$ or be approximated arbitrarily well by a random dot
product graph with growing $d$.  For the random dot product graphs,
\cite{tang2017semiparametric} considered the two-sample problem of
determining whether or not two graphs on the same vertex set have the
same generating latent positions or have generating latent positions
that are scaled or diagonal transformations of one another;
\cite{tang2017nonparametric} studied a related problem wherein the
graphs can have different vertex sets with possibly differing numbers
of vertices. Another special case of the latent space model specifies
that $f$ is the logit link function and for this choice of $f$,
\cite{durante2018bayesian} developed a Bayesian procedure for testing
group differences in the network structure that also relies on a
low-rank representation of the latent positions together with
edge-specific latent covariates.

The appeal, and consequently power and utility,
of the latent space formulation for two-sample testing
stems from the fact that a $n$ vertices
graph can be parameterized by the $n \times d$ matrix
of latent positions $\{x_i\}$; this is, when $n \gg
d$, a considerable reduction in the number of parameters compared
to the $n(n-1)/2$ edge probabilities. This
reduction, however, is possible only if the $\{x_i\}$ can be estimated accurately, and this is generally done by assuming that $f$ is known; e.g., $f$ is a bilinear form 
\cite{tang2017semiparametric,tang2017nonparametric} or the
logistic function \cite{durante2018bayesian}. 

We consider in this paper another two-sample testing problem for
latent position graphs, but, in contrasts to existing works
we neither assume that the link function $f$ is known nor that it need
to be the same between the two-samples. More specifically, we consider
the class of latent distance random graphs wherein we assume that
$f(x_i, x_j) = h(\|x_i - x_j\|)$ for some {\em unknown} non-increasing
function $h$ that could differ between the two samples. It is not
{\em a priori} clear that the latent positions are even identifiable;
we show subsequently that the latent positions are identifiable up to
a similarity transformation.

The problem is of significant theoretical and practical interest
because of the following reasons. The first is that many of the
currently studied two-sample testing problems have test statistics
that are constructed using the difference of adjacency matrices or the
estimated edge probability matrices.  Since we assume the link
function is unknown and possibly different, this commonly used method
is no longer valid. As we will clarify later, even when we know that
the link functions are of the same form, they may still depend on
unknown parameters that are different between the two samples, and
thus we have different edge-probability matrices which cannot be
compared directly. The second reason is that, due to the
non-identifiability of latent positions, our test procedures allow for
more flexible comparisons than just whether or not the two latent
positions are {\em exactly} the same, i.e., our tests are for equality
up to general similarity transformation which includes any
transformation that preserves the {\em ordering} of pairwise
distances.

Our test procedure, even after accounting for all this complex source
of non-identifiability, is quite simple. We estimate the
edge probabilities matrices by truncating the singular value
decomposition of the averaged adjacency matrices in each population and
then compute our test statistic as the Spearman rank correlation between these
estimates. Significance values are obtained either via a permutation
test when the number of samples from each population is moderate, or
via a bootstrapping scheme in the case when there are only one or two
samples in each population.

\section{Methodology}
\label{sec:methodology}
We first recall the definition of latent distance random graphs \cite{hoff2002latent}.
\begin{definition}[Latent Distance Random Graphs]
Let $h$ be a monotone decreasing function from $\mathbb{R}$ to $[0,1]$ and assume $h(0)=1$ for identifiability. Let $n,d \geq 1$ be given and let $X = [x_1 \mid \dots \mid x_n]^{\top}$ be a $n \times d$ matrix with rows $x_i \in \mathbb{R}^{d}$. A $n \times n$ adjacency matrix $A$ is said to be an instance of a latent distance random graph with latent position $X$ and sparsity parameter $\rho \in (0,1]$ if $A$ is a symmetric, hollow matrix whose upper triangular entries $a_{ij}, i < j$ are conditionally independent Bernoulli random variables with $\mathrm{pr}(a_{ij} = 1) = \rho h(\|x_i - x_j\|)$, i.e., the likelihood of $A$ given $X$ is
	\begin{align*}
		\mathrm{pr}(A|X)=\prod_{i<j}\left\{\rho h(\|{x}_i-{x}_j\|)\right\}^{a_{ij}}\left\{1-\rho h(\|{x}_i-{x}_j\|)\right\}^{1-a_{ij}}.
	\end{align*}


\end{definition}

The graphs we study are undirected, unweighted and loop-free. Given
two adjacency matrices $A$ and $B$ for a pair of random latent
distance graphs on the same set of vertices, we will propose a valid,
consistent test to determine whether the two generating latent
positions are equal up to similarity transformation, e.g., scaling and
orthogonal transformation.

Generally speaking, the link functions are unknown. 
Even when the specific form of the link functions are known, there could still be unknown parameters. 
For example, the original latent space model of \cite{hoff2002latent} uses the logistic function, i.e.,
$$h(\|{x}_i-{x}_j\|)=\frac{\exp(\alpha-\beta\|{x}_i-{x}_j\|)}{1+\exp(\alpha-\beta\|{x}_i-{x}_j\|)},$$where $\alpha\in\mathbb{R}$ and $\beta>0$.
Other alternatives were discussed in \cite{raftery2017comment}. \cite{gollini2016joint} replaced the Euclidean distance by the squared distance to allow higher edge-probability for close points. \cite{rastelli2016properties} replaced the logistic function by a Gaussian kernel
$$h(\|{x}_i-{x}_j\|)=\gamma\exp\left(-\frac{\|{x}_i-{x}_j\|^2}{2\phi}\right),$$
where $\gamma\in[0,1]$ and $\phi>0$. Even if the link functions of networks to be compared are known and in the same form, it is still reasonable for them to have different parameters, such as $\alpha,\beta$ in the logistic function and $\gamma,\phi$ in the Gaussian function. Furthermore, even when the link functions are  the same with exactly identical parameters, if the latent positions are similar up to an unknown similarity transformation then the edge-probability matrices are not equal and cannot be compared directly. We are thus motivated to consider the following two-sample hypothesis testing problem.

Let $X,Y\in\mathbb{R}^{n\times d}$. We define the edge-probability
matrices $P=(p_{ij})\in\mathbb{R}^{n\times n}$ where
$p_{ij}=h(\|{x}_i-{x}_j\|)$ and $Q=(q_{ij})\in\mathbb{R}^{n\times n}$
where $q_{ij}=g(\|{y}_i-{y}_j\|)$. The link functions $h$ and $g$ are
unknown and possibly different. We shall assume, for identifiability, that $h(0) = g(0) = 1$. Given $A_1,...,A_m$ and $B_1,...,B_m$ generated from latent distance random graphs with latent positions $X$ and $Y$ respectively, where $m\geq1$, the two-sample testing problem is defined formally as 
\begin{align*}
&H_0: X=sYW+1t^{\top}\text{ for some $s\in\mathbb{R}$, orthogonal $W\in\mathbb{R}^{d\times d}$, ${t}\in\mathbb{R}^d$ against}\\
&H_a:X\neq sYW+1t^{\top}\text{ for any $s\in\mathbb{R}$,  orthogonal $W\in\mathbb{R}^{d\times d}$, ${t}\in\mathbb{R}^d$}.
\end{align*}
The above null hypothesis captures the notion that two latent
positions are the same up to similarity transformations.

Our test procedure starts by estimating the edge-probability matrices using a
singular value thresholding procedure. More specifically, we compute $\bar{A} =
m^{-1} \sum_{i} A_i$ and let $\hat{P}$ be the best rank-$K$ approximation of $\bar{A}$ with respect
to the Frobenius norm, i.e., $\hat{P}$ is obtained by computing the
singular value decomposition of $\bar{A}$ and keeping only the $K$
largest singular values and corresponding singular vectors. The estimate $\hat{Q}$ of $Q$ is constructed similarly. Singular value
thresholding procedures have been actively studied in \cite{chatterjee2015matrix}
and \cite{xu2017rates}. As discussed in \cite{xu2017rates}, the
choice of dimension $K$ can be determined by a threshold
$\tau=c_0(n\rho)^{1/2}$ where $c_0$ is a universal constant
strictly larger than $4$ in the case of $n\rho\gg\log(n)$ and strictly
larger than $2$ in the case of $n\rho\gg\log^4(n)$. For our simulation
and real data analysis, we frequently set $K$ to a fixed value or
choose $K$ using the dimension selection procedure of \cite{zhu_ghodsi}.

Since the link function is assumed to be monotone in latent distance
random graphs, under the null hypothesis, the ordering of the entries
in the two edge-probability matrices should be the same, i.e., $p_{ij} \leq
p_{k \ell}$ if and only if $q_{ij} \leq q_{k \ell}$. Thus, several
rank-based or order-based methods can be used to construct similar test
statistic, including Kendall's $\tau$ coefficient, Spearman's rank
correlation coefficient, non-metric multidimensional scaling and
isotonic regression. There are, however, important computational or
theoretical challenges for some of these methods. In particular, Kendall's $\tau$ is computationally intensive with $O(n^4)$
time complexity where $n$ is the number of vertices. An approximation
for Kendall's $\tau$ with $O(n^2\log n)$ complexity has been developed
but its impact on the theoretical properties of the resulting test statistic is
unknown. Non-metric multidimensional scaling is also computationally
intensive as it is generally formulated as a non-convex problem with
multiple local minima and thus one is not guaranteed to find the
global minimum. Using isotonic regression, we can consider
$p_{ij}=f(q_{ij})$ and test whether the function
$f$ is monotone, but the corresponding theory in the case where the
predictor variable is noisy has not been well-studied.

We thus propose a test statistic defined using Spearman's rank
correlation coefficient, which is computationally efficient with
$O(n^2)$ complexity. That is,
\begin{align}
\label{test1}
T_n(\hat{P},\hat{Q})&=\frac{\mathrm{cov}\{R(\hat{P}),R(\hat{Q})\}}{\hat{\sigma}\{R(\hat{P})\}\hat{\sigma}\{R(\hat{Q})\}},
\end{align} where $R(\hat{P})\in\mathbb{R}^{n\times n}$ and
$R(\hat{Q})\in\mathbb{R}^{n\times n}$ are symmetric matrices whose
entries are the ranks of the corresponding entries in $\hat{P}$ and $\hat{Q}$, 
$\mathrm{cov}\{R(\hat{P}),R(\hat{Q})\}$ is the sample covariance of
these ranks, and $\hat{\sigma}\{R(\hat{P})\}$ and
$\hat{\sigma}\{R(\hat{Q})\}$ are the standard
deviations. Given the significance level $\alpha\in(0,1)$, the rejection region
$\mathcal{R}$ for the test statistic $T_n$ is $\mathcal{R}=\left\{t
|\text{ $p$-value}(t)<\alpha\right\}$, where $\text{$p$-value}(t)$ can
be determined either via a permutation test or via a bootstrapping
procedure as described in Algorithm
\ref{alg_boot}. When the number of samples from each population is moderately
large then both methods should perform well. If
the number of samples is small, or even in the case when there is only
a single network observation for each population, then the bootstrap could be more robust.

\begin{algorithm}[tp]
\caption{Bootstrap}
\label{alg_boot}
\textbf{Input}: Two adjacency matrices $A,B\in \mathbb{R}^{n\times n}$. Number of bootstrap replications $N$.
\begin{enumerate} \item[Step 1.] Apply universal singular value thresholding on $A$ and $B$ to get $\hat{P}$ and $\hat{Q}$.
	\item[Step 2.] Calculate test statistic as $t^*:=T_n(\hat{P},\hat{Q})$.
\end{enumerate}
For $k=1,...,N$, repeat steps 3, 4 and 5:
\begin{enumerate}
 	\setcounter{enumi}{2}
	\item[Step 3.] Generate $A_1^{(k)}=\left(a_{1,ij}^{(k)}\right)_{n\times n}$, $A_2^{(k)}=\left(a_{2,ij}^{(k)}\right)_{n\times n}$, $B_1^{(k)}=\left(b_{1,ij}^{(k)}\right)_{n\times n}$ and $B_2^{(k)}=\left(b_{2,ij}^{(k)}\right)_{n\times n}$ as
	\begin{gather*}
	    a_{1,ij}^{(k)}\overset{i.i.d}{\sim} \text{Bernoulli}(\hat{p}_{ij}),\quad a_{2,ij}^{(k)}\overset{i.i.d}{\sim} \text{Bernoulli}(\hat{p}_{ij}), \\
	    b_{1,ij}^{(k)}\overset{i.i.d}{\sim}\text{Bernoulli}(\hat{q}_{ij}),\quad b_{2,ij}^{(k)}\overset{i.i.d}{\sim} \text{Bernoulli}(\hat{q}_{ij}).
	    \end{gather*}
	\item[Step 4.] Apply universal singular value thresholding on the bootstrapped adjacency matrices to get $\hat{P}_1^{(k)}$, $\hat{P}_2^{(k)}$, $\hat{Q}_1^{(k)}$ and $\hat{Q}_2^{(k)}$.
	\item[Step 5.] Calculate test statistic as $t_P^{(k)}:=T_n\left(\hat{P}_1^{(k)},\hat{P}_2^{(k)}\right)$ and $t_Q^{(k)}:=T_n\left(\hat{Q}_1^{(k)},\hat{Q}_2^{(k)}\right)$ 
	\item[Step 6.] Calculate the $p$-value as $$\text{$p$-value}=\min\left[\max\left\{\frac{1}{N}\sum_{k=1}^NI\left(t^*<t_P^{(k)}\right),\frac{1}{N}\sum_{k=1}^NI\left(t^*<t_Q^{(k)}\right)\right\},1\right].$$
\end{enumerate}
\textbf{Output}: $p$-value of the proposed testing procedure.
\end{algorithm}

\section{Main Results} We now establish the main theoretical properties of the
proposed test procedure as the number of
vertices $n$ increases. We will assume that the number of graphs $m$ in each sample, is bounded; thus, for ease
of exposition, we set $m = 1$ throughout. The case when $m
\rightarrow \infty$ with increasing $n$ is considerably simpler and is
thus ignored.

We start by introducing several mild assumptions on the
link functions $h$ and $g$ and the latent positions $\{x_i\}_{i=1}^{n}$ and
$\{y_i\}_{i=1}^{n}$.
\begin{assumption}
As $n\to\infty$, for any $\epsilon>0$, there exists $\delta>0$ such that $[0,1)$ can be partitioned into the union of intervals of the form $[(k-1)\delta,k\delta)$ for $k=1,...,\lceil 1/\delta \rceil$, such that, for any $k$, one of the following two conditions holds almost surely:
\begin{enumerate}[label=(\roman*)]
	\item Either the number of $ij$ pairs with $i < j$ and $p_{ij}\in[(k-1)\delta,k\delta)$ is at most $n(n-1)\epsilon/2$.
	\item Or if the number of $ij$ pairs with $i < j$ and $p_{ij}\in[(k-1)\delta,k\delta)$ exceeds $n(n-1)\epsilon/2$, then they are all equal for $p_{ij} \in [(k-1)\delta,k\delta)$.
\end{enumerate}
\end{assumption}
\begin{assumption}
Define $$\hat{\sigma}\{R(P)\}=\Bigl[\tbinom{n}{2}^{-1}\sum_{i<j}\Bigl\{R(p_{ij})-\tbinom{n}{2}^{-1}\sum_{i<j}R(p_{ij})\Bigr\}^2\Bigr]^{1/2}$$ as the sample variance for the ranks of the entries in $P$. Define $\hat{\sigma}\{R(Q)\}$ similarly. Then as $n \rightarrow \infty$, 
$\hat{\sigma}\{R(P)\}=\Omega(n^2)$ and $\hat{\sigma}\{R(Q)\}=\Omega(n^2)$ almost surely.
\end{assumption}
\begin{assumption}
The link functions $h$ and $g$ are fixed with $n$ and both are
infinitely differentiable. The latent positions $x_i \in U \subset
\mathbb{R}^{d} $and $y_i \in V \subset \mathbb{R}^{d}$ for some fixed
compact sets $U$ and $V$ that do not depend on $n$.
\end{assumption}
\begin{assumption}
There exists a constant $C$ not depending on $n$ such that, as $n$
increases, the sparsity parameter $\rho \in(0,1]$ satisfies $n\rho\geq
C\log n$.
\end{assumption}

\begin{remark} We now explain the rationale behind the above
  assumptions.
  \begin{enumerate}[label=(\roman*)]
    \item Assumption 1 prevents the setting where a large number of
latent positions concentrate around a single point $x_0$ with
increasing $n$ but that these points are not equal to $x_0$. If this
happens then the values of the $p_{ij}$ for this collection of points
would be almost identical but their ranks are substantially
different. For example, suppose there are $cn$ points around a small
neighbourhood of $x_0$. Then the $p_{ij}$ for the $cn(cn-1)/2$ pairs
in this neighbourhood will all be approximately $h(0)$. The rank of
the smallest and the largest of these $p_{ij}$ could, however, differ
by $cn(cn-1)/2$. Assumption $1$ arises purely because we do not assume
anything about a generative model for the latent positions
$\{x_i\}$. Indeed, if the latent positions $x_i$ are independent and
identically distributed samples from some distribution $F$, then for
any point $x_0$, either $F$ has an atom at $x_0$ which will then force
$p_{ij} = c$ for some constant $c$ whenever $x_i = x_j =
x_0$. Otherwise, if $F$ is non-atomic at $x_0$ then the proportion of
points $x_i$ with $\|x_i - x_0\| \leq \delta$ will converge to $0$ as
$\delta \rightarrow 0$.
    \item Assumption 2 complements Assumption 1 and prevents the ranks
of the entries of $P$ and $Q$ from being degenerate. Suppose, for
example, that there are $n-o(n)$ points located at a certain position
${x}_0$. In this case, $\hat{\sigma}\{R(P)\}=o(n^2)$. The problem of
testing whether $X$ is equal to $Y$ up to a similarity transformation
can thus be reduced to consider only the subgraphs induced by these
$o(n)$ points. We can then apply the test procedure in this paper,
assuming that these induced subgraphs can be found efficiently. The
problem of identifying these subgraphs is, however, outside the scope
of our current investigation.
    \item Assumption 3 restricts the smoothness of the link
functions. This is done entirely for ease of exposition. The
assumption can easily be relaxed as it only affects the accuracy of
the universal singular value thresholding procedure used in estimating
edge-probability, which in turn affects the convergence rate of the
test statistic. More specifically, suppose the link function $h$
belongs to a H\"{o}lder class or Sobolev class with index
$\omega$. Then Theorem 1 of \cite{xu2017rates} implies
that $$\frac{1}{n^2}\|\hat{P}-P\|_F^2=O_
p\left((n\rho)^{-\frac{2\omega}{2\omega+d}}\right).$$ The convergence
rate of our test statistic will then depends on $\omega$ and is thus
slower than the convergence rate for infinitely differentiable link
functions as given in Corollary~\ref{cor:main}.
  \item Assumption~4 is identical to that used in
    \cite{xu2017rates}. We restrict the sparsity of the observed
graphs in order to guarantee that the singular value thresholding
estimates $\hat{P}$ and $\hat{Q}$ are accurate estimates of $P$ and
$Q$.
   \end{enumerate}
\end{remark}

With the above assumptions in place, we now show that the test
statistic $T_n$ constructed using appropriately {\em discretized}
versions of the estimated edge
probabilities matrices $\hat{P}$ and $\hat{Q}$ is, asymptotically, the
same as that constructed using the true $P$ and $Q$. The need for
discretizing the entries of $\hat{P}$ and $\hat{Q}$ is due mainly to
the fact that the estimates $\{\hat{p}_{ij}\}$ and $\{\hat{q}_{ij}\}$
are inherently noisy. Suppose for example that $p_{ij}=0.1$ for all $ij$ pairs. Then
$R(p_{ij}) \equiv \{n(n-1)/2 + 1\}/2$, the average of the ranks in
$\{1,...,n(n-1)/2\}$. However, because of the estimation error, the
$\hat{p}_{ij}$ might contain numerous distinct values like
$\{0.101,0.102,...\}$ and thus
$\sum_{i<j}\{R(p_{ij})-R(\hat{p}_{ij})\}^2$ can be quite large even
though the estimates $\hat{p}_{ij}$ are all approximately equal to the
true $p_{ij}$. We are thus motivated to consider a more
robust estimator obtained by discretizing the $\{\hat{p}_{ij}\}$,
i.e., let $\eta > 0$ and define the $\eta$-discretization of $\hat{p}_{ij}$ as
$$\tilde{p}_{ij}=\ceil* {\frac{\hat{p}_{ij}}{\eta}} \times \eta.$$
Recall the above example. By letting $\eta=0.01$, we have
$\tilde{p}_{ij} = 0.1$ provided that $\mid\hat{p}_{ij} - 0.1\mid \leq
0.01$ and hence the ranks of these $\tilde{p}_{ij}$ are the
same. We emphasize
that while this discretization step simplifies the subsequent theory
considerably, it is not essential in real data analysis as we can
always choose $\eta$ sufficiently small so that $\tilde{p}_{ij}$ is
arbitrarily close to $\hat{p}_{ij}$.

\begin{theorem}
  \label{thm:main_theorem}
  Assume Assumptions 1--4 hold. Then for sufficiently
large $n$,
$$T_n(\tilde{P},\tilde{Q})-T_n(P,Q)=o_p(1).$$
Here $\tilde{P}$ and $\tilde{Q}$ are the $\eta$-discretization of
$\hat{P}$ and $\hat{Q}$ with $(\eta^2 \rho)^{-1} = o(n)$ as $n \rightarrow \infty$. 
\end{theorem}
Theorem~\ref{thm:main_theorem} indicates that
$T_n(\tilde{P},\tilde{Q})-T_n(P,Q)$ converges to $0$ as $n \rightarrow
\infty$. There are, however, instances in which we are interested
in the rate of convergence of $T_n(\tilde{P},\tilde{Q})-T_n(P,Q)$ to $0$. We
derive the rate of convergence under the following more restrictive version
of Assumption~1.
\begin{assumption} There exists a constant $c>0$ independent of
$\delta$ such that for $k = 1,\dots,\lceil 1/\delta \rceil$, 
$$\mid\{(i,j):p_{ij}\in [(k-1)\delta,k\delta]\}\mid\leq c\cdot \delta \tbinom{n}{2}.$$
\end{assumption}
\begin{corollary}
\label{cor:main}
Under the conditions in Theorem 3.1 and Assumption 5, we have
$$T_n(\tilde{P},\tilde{Q})-T_n(P,Q)=O_p(\epsilon^{1/2}),$$
where $(\epsilon^2\rho)^{-1}=o(n)$.
\end{corollary} Assumption~5 allows us to set $\eta = \epsilon = c
\delta$ in the proof of Theorem~\ref{thm:main_theorem}, thereby
yielding the convergence rate of $O_{p}(\epsilon^{1/2})$ in
Corollary~1. If the conditions in Assumption~5 are not satisfied then
there is, {\em a priori}, no explicit relationship
between $\eta$ and $\epsilon$ other than that $\epsilon \to 0$ as
$\eta\to 0$.

If the null hypothesis is true then $T_n(P, Q) = 1$ and hence, from
Theorem~\ref{thm:main_theorem}, we have $T_n(\tilde{P}, \tilde{Q})
\rightarrow 1$ almost surely as $n \rightarrow \infty$. A natural
question then is whether or not $T_n(\tilde{P},\tilde{Q}) \rightarrow
1$ also indicates that the matrix of latent positions $X$ is close, up to some similarity
transformation, to the matrix of latent positions $Y$ ? To address
this question we shall assume that the sequences of latent
positions $\{x_i\}_{i=1}^{n}$ and $\{y_i\}_{i=1}^{n}$ satisfy the
following denseness conditions as $n \rightarrow \infty$.
\begin{assumption}
  \label{ass:dense}
Let $U \subset\mathbb{R}^d$ and $V \subset \mathbb{R}^{d}$ be
non-empty, bounded and connected sets. Let $\Omega_n=\{{x}_1,...,{x}_n\}\subset U$ and
  $\Xi_{n} = \{y_1, y_2, \dots, y_n\} \subset V$. Then $\lim_{n
    \rightarrow \infty} \Omega_n$ and $\lim_{n \rightarrow \infty} \Xi_n$ 
  are dense in $U$ and $V$, respectively. Furthermore, for any $\epsilon > 0$
  there exist $\delta_U = \delta_U(\epsilon) > 0$ and $\delta_V =
  \delta_V(\epsilon) > 0$ depending on $\epsilon$ such that
  \begin{gather*}
  n^{-1} \liminf |B(x,\epsilon) \cap \Omega_n| \geq \delta_U, \qquad
  \text{for all $x \in U$}, \\
  n^{-1}\liminf |B(y,\epsilon) \cap \Xi_n| \geq \delta_V, \qquad
  \text{for all $y \in V$}.
  \end{gather*}
  Here $B(x, \epsilon)$ denote the ball of radius $\epsilon$ centered
  at $x$.
\end{assumption}
Assumption~\ref{ass:dense} is a regularity condition for the
minimum number of latent positions $\{x_i\}$ and $\{y_i\}$ in any arbitrarily
small, but non-vanishing subset of $U$ and
$V$. In particular, Assumption~\ref{ass:dense} prevents the setting
where, as $n \rightarrow \infty$, the sequence of latent positions
$\{x_i\}_{i=1}^{n}$ is dense in $U$, but that, for any sufficiently
large $n$, all except
$o(n)$ of these
$\{x_i\}$ are concentrated at some fixed
$K$ points $\nu_1, \dots \nu_K \in U$, i.e., the
denseness of the $\{x_i\}_{i=1}^{n}$ is due to a vanishing fraction of
the points. While the removal of these $o(n)$ points from both
$\{x_i\}$ and $\{y_i\}$ does not change the convergence $T_n(P, Q)$ to
$1$, it will lead to very different geometry for the remaining
latent positions. In summary, as we only require $T_n(P,Q) \rightarrow
1$, Assumption~\ref{ass:dense} guarantees
that the removal of any $o(n)$ points from the $\{x_i\}$ and
$\{y_i\}$ does not substantially change the geometry of the remaining
points, especially since the removal of any $o(n)$ points does not
change the convergence of $T_n(P,Q)$. 

The following result showed that if $X$ and $Y$ satisfy the conditions in Assumption~\ref{ass:dense}
and $T_n(P,Q) \rightarrow 1$ as $n
\rightarrow \infty$ then the Frobenius norm distance between $X$
and some similarity transformation of $Y$ is of
order $o(n^{1/2})$. Since there are $n$ rows in $X$ and $Y$, this indicates that for any {\em arbitrary}
but fixed $\epsilon > 0$, the number of rows $i$ such that $\|X_i - s
W Y_i - t\| \geq \epsilon$ is of order $o(n)$ as $n$ increases. That
is to say, almost all rows of $X$ are arbitrarily close to the
corresponding rows of some similarity transformation of
$Y$. %

\begin{theorem}
  \label{thm:converse}
Suppose that, as $n \rightarrow \infty$, the latent positions $X$ and
$Y$ satisfy Assumption $1$ through $4$ together with
Assumption~$6$. If $T_n(P,Q)\to 1$ as $n\to\infty$ then there exists $s
\in\mathbb{R}$, orthogonal matrix $W\in\mathbb{R}^{d\times d}$ and $t \in\mathbb{R}^d$ such that 
$$\|X-sYW-1t^{\top}\|_F=o(n^{1/2}).$$
\end{theorem}
The detailed proofs of Theorem~\ref{thm:main_theorem} and Theorem~\ref{thm:converse} are given in the
supplementary materials.  

We finally discuss the consistency of our test procedure. Since, as
$n$ increases, the dimension of our latent positions and the
associated edge probabilities matrices also increases, we shall define
consistency of our test procedure in the context of a sequence of
hypothesis tests.
\begin{definition}[Consistency]
  \label{def:consistency}
Let $(X_n,Y_n)_{n\in \mathbb{N}}$ be a given sequence of latent positions, where $X_n$ and $Y_n$ are both in $\mathbb{R}^{n\times d}$. A test statistic $T_n$ and associated rejection region $\mathcal{R}$ to test the hypothesis
\begin{equation*}
  \begin{split}
&H_0: X=sYW+1t^{\top}\text{ for some $s\in\mathbb{R}$, orthogonal $W\in\mathbb{R}^{d\times d}$, ${t}\in\mathbb{R}^d$ against}\\
&H_a:X\neq sYW+1t^{\top}\text{ for any $s\in\mathbb{R}$,  orthogonal $W\in\mathbb{R}^{d\times d}$, ${t}\in\mathbb{R}^d$}.
  \end{split}
\end{equation*}
is a consistent, asymptotically level $\alpha$ test if for any $\epsilon>0$, there exists $n_0=n_0(\epsilon)$ such that:
\begin{enumerate}[label=(\roman*)]
    \item If $n>n_0$ and $H_a^n$ is true, then $\mathrm{pr}(T_n\in \mathcal{R})>1-\epsilon$.
    \item If $n>n_0$ and $H_0^n$ is true, then $\mathrm{pr}(T_n\in \mathcal{R})\leq\alpha-\epsilon$.
\end{enumerate}
\end{definition}
\begin{theorem}
\label{thm:consistency}
Let $\{X_n\}_{n \geq 1}$ and $\{Y_n\}_{n \geq 1}$ be two sequences of
matrices of latent positions for the latent position graphs with link functions $g$ and $h$, respectively. Suppose that, as $n \rightarrow
\infty$, these latent positions and associated link functions satisfy Assumptions $1$ through $4$ together with
Assumption~$6$. 
For each fixed $n$, consider the hypothesis test in Definition~\ref{def:consistency}
for the $X_n$ and $Y_n$. Define the test
statistic $T_n(\tilde{P},\tilde{Q})$ as in Eq.\eqref{test1}. Let $\alpha\in(0,1)$ be given. If the rejection
region is $\mathcal{R}=\{t\in\mathbb{R}:t< C\}$ for some constant
$C\leq 1$, then there exists an $n_0=n_0(\alpha,\epsilon)\in
\mathbb{N}$ such that for all $n\geq n_0$, the test procedure with
$T_n$ and the rejection region $\mathcal{R}$ is an at most level
$\alpha$ test, that is, if the null hypothesis $H_0$ is true, then
$\mathrm{pr}(T_n\in\mathcal{R})\leq \alpha-\epsilon$. Denote by
$$d_n= \min_{s, W, t} \|X_n-sY_nW-1t^{\top}\|_F$$
the minimum Frobenius norm distance, up to some similarity transformation, between $X_n$ and $Y_n$.
Then the test procedure is consistent in the sense of Definition 2 over this
sequence of latent positions if, as $n \rightarrow \infty$, $\liminf
n^{-1/2} d_nI\{d_n>0\}>0$ where $I(\cdot)$ is the indicator function.
\end{theorem}
\begin{remark}
  In Theorem~\ref{thm:consistency}, $\alpha$ need not depend on $C$ since we have not derived a
  non-degenerate limiting distribution for our test
  statistic. Theorem~\ref{thm:consistency} indicates that, for
  sufficiently large $n$, our test procedure has power arbitrarily
  close to $1$ whenever the minimum Frobenius
  norm distance between $X_n$ and any similarity transformation of $Y_n$ is of order $\Omega(n^{1/2})$. Thus, roughly speaking, the test
  procedure has power converging to $1$ if there does not exists a similarity
  transformation mapping the rows of $X_n$ to that of $Y_n$; see the discussion prior to the statement of Theorem~\ref{thm:converse}.
\end{remark}

\section{Simulations}
\subsection{General Procedure}
We first summarize the setup and general procedure used for generating the empirical distributions of our test statistic.
\begin{enumerate}[label=(\alph*)]
	\item For $i=1,...,n$ and $j=1,2$, generate $x_{ij}\overset{iid}{\sim} N(0,1)$ and form $X=(x_{ij})\in\mathbb{R}^{n\times 2}$. 
	\item We set different $Y$ for null and alternative hypotheses:
	\begin{itemize}
		\item Under $H_0$: Set $Y=(1+\epsilon)X$.
		\item Under $H_a$: Set $Y=X+Z$ where $Z=(z_{ij})\in\mathbb{R}^{n\times 2}$ and $z_{ij}\overset{iid}{\sim}  N(0,\epsilon)$ independent from $x_{ij}$.
	\end{itemize}
	\item The edge-probability matrices based on $X$ and $Y$ are respectively defined as $P=(p_{ij})\in\mathbb{R}^{n\times n}$ and $Q=(q_{ij})\in\mathbb{R}^{n\times n}$, where
		\begin{align*}
			p_{ij}=h({x}_i,{x}_j)&=\exp(-\|{x}_i-{x}_j\|^2), \qquad
			q_{ij}=g({y}_i,{y}_j)&=\exp(-\|{y}_i-{y}_j\|^2/4).
		\end{align*}
	\item Generate the corresponding adjacency matrices $A$ and $B$ as $A_{ii}=B_{ii}=0$ for $i=1,...,n$ and $A_{ij}=\text{Bernoulli}(\rho p_{ij})$, $B_{ij}=\text{Bernoulli}(\rho q_{ij})$ for $i,j=1,...,n$ and $i\neq j$.
	\item Apply universal singular value thresholding on $A$ and $B$ to get the estimates of $P$ and $Q$ as $\hat{P}$ and $\hat{Q}$.
	\item Calculate the test statistic $T_n(\hat{P},\hat{Q})$.
	\item Repeat (d)-(f) 100 times or use other resampling techniques to get the empirical distribution of $T_n(\hat{P},\hat{Q})$.
\end{enumerate}
\subsection{Experiments} We first show that our proposed test
procedure exhibits power for small and moderate values of $n$ in
Simulation 1. We then study the performance of the permutation test
and bootstrap procedure in Simulation 2. Finally we compare our test
procedure with another procedure that is based on non-metric embedding
of the adjacency matrices. An additional simulation on sparsity and
its effects on our test procedure is included in the supplementary
materials.

\textbf{\textit{Simulation 1: Power.}} This simulation is designed to
investigate power of the proposed test as the number of
vertices vary and for different settings of the latent positions. Set $K = 3$ in the singular
value thresholding procedure, sparsity level $\rho=1$,
$n\in\{50,100,200,500,1000\}$ and significant level $\alpha
=0.05$. Recall the two settings of latent positions are
\begin{itemize}
    \item $M_1$: $Y=(1+\epsilon)X$.
    \item $M_2$: $Y=X+Z$ where $Z=(z_{ij})\in\mathbb{R}^{n\times 2}$ and $z_{ij}\overset{iid}{\sim} N(0,\epsilon)$ independent from $x_{ij}$.
\end{itemize} Set $\epsilon\in\{0,0.02,0.1,0.2,0.5\}$ and note that
the null hypothesis is true under $M_1$ for all values of
$\epsilon$. In contrast, the null hypothesis is true under $M_2$ if
and only if $\epsilon = 0$. The power for
different settings, reported in Table \ref{powertable}, is calculated
based on the empirical distribution generated by the procedure
outlined in Section 4.1.
\begin{table}
\centering
\def~{\hphantom{0}}
\caption{Power of the proposed test ($\alpha=0.05$).}{
\begin{tabular}{ccccccc}
  Setting&$n$&$\epsilon=0$&$\epsilon=0.02$&$\epsilon=0.1$&$\epsilon=0.2$&$\epsilon=0.5$\\
  \multirow{5}{*}{$M_1$}&50&0.05&0.04&0.02&0&0\\
  &100&0.05&0.04&0&0&0\\
  &200&0.05&0.02&0&0&0\\
  &500&0.05&0&0&0&0\\
  &1000&0.05&0&0&0&0\\
  \multirow{5}{*}{$M_2$}&50&0.05&0.35&1&1&1\\
  &100&0.05&1&1&1&1\\
  &200&0.05&1&1&1&1\\
  &500&0.05&1&1&1&1\\
  &1000&0.05&1&1&1&1\\
 \end{tabular}
 }
\label{powertable}
\end{table}

According to Table 1, our testing procedure is valid. As $M_1$
satisfies the null hypothesis, the power of the proposed test is
approximately $0$. While the test appears to be slightly conservative,
this is due mainly to the fact that the edge-probability matrices for
the two samples are quite different when $\epsilon > 0$, e.g., when
$\epsilon = 0.1$ and $n=50$ the average edge density for the two
populations are 0.23 and 0.40. This difference impacts the
finite-sample estimation $\hat{P}$ and $\hat{Q}$ and the resulting
test statistic $T_n(\hat{P}, \hat{Q})$. The test procedure also
exhibits power even for small values of $\epsilon$ and moderate values
of $n$, e.g., for the setting $M_2$ we see that
the empirical power of the proposed test is $1$ except for $\epsilon=0.02$ and $n=50$ where the
difference between the latent positions is miniscule and the sample size
is small.

\textbf{\textit{Simulation 2: Permutation test and bootstrap.}}
Simulation 1 simply resampled data from the distribution under the
null hypothesis. This is appropriate in simulation studies but does
not yield a valid test procedure in practice. To get a valid test procedure, we
consider other resampling techniques. This simulation
is designed to understand the performance of permutation test and
bootstrap procedure in our test. We set $K = 3$ in the singular value
thresholding procedure and set the sparsity level $\rho=1$. The latent positions are set to be $Y=X+Z$
where $Z=(z_{ij})\in\mathbb{R}^{n\times 2}$ and
$z_{ij}\overset{iid}{\sim} N(0,\epsilon)$ and
$\epsilon\in\{0,0.02,0.1,0.2,0.5,1\}$.

For the permutation test, each sample consists of 100 adjacency matrices from that population. We apply discretization on the estimated edge-probability matrices with $\eta=0.05$. Meanwhile for the bootstrap test, each sample consists of a single graph from that population.
\begin{table}
\centering
\def~{\hphantom{0}}
\caption{Power of 100 replications of 1000 Permutation Test ($\alpha=0.05$).}{
\begin{tabular}{ccccccc}
  $n$&$\epsilon=0$&$\epsilon=0.02$&$\epsilon=0.1$&$\epsilon=0.2$&$\epsilon=0.5$&$\epsilon=1$\\
  100&0&0&1&1&1&1\\
  200&0&0&1&1&1&1\\
  500&0&0&1&1&1&1\\
  \end{tabular}
  }
  \label{tbl:permutation}
\end{table}
\raggedbottom
\begin{table}
\centering
\def~{\hphantom{0}}
\caption{Power of 100 replications of 1000 Bootstrapping ($\alpha=0.05$).}{
\begin{tabular}{ccccccc}
  $n$&$\epsilon=0$&$\epsilon=0.02$&$\epsilon=0.1$&$\epsilon=0.2$&$\epsilon=0.5$&$\epsilon=1$\\
  20&0.11&0.12&0.26&0.44&0.67&0.87\\
  50&0&0.01&0.25&0.58&1&1\\
  100&0.01&0&0.97&1&1&1\\
  200&0&0.43&1&1&1&1\\
  \end{tabular}
}
  \label{tbl:bootstrap}
\end{table}

Tables~\ref{tbl:permutation} and \ref{tbl:bootstrap} show that the permutation test and the bootstrapping procedure both exhibit power even for small values of $n$, provided that the discrepancy between the latent positions as captured by $\epsilon$ is not too small. Indeed, even though both approaches are conservative for $n \geq 100$, the power of the test is approximately $1$ for all $\epsilon \geq 0.1$. 


\textbf{\textit{Simulation 3: Comparison with non-metric multidimensional scaling.}} We next perform a simulation study to compare our test procedure with the test procedure in \cite{xixihu} that is based on embedding the adjacency matrices via non-metric multidimensional scaling. More specifically, given a $n \times n$ {\em weighted} adjacency matrix $A$ and an embedding dimension $d$, non-metric multidimensional scaling seeks to find a collection of points $x_1, \dots, x_n$ in $\mathbb{R}^{d}$ such that the pairwise distances between the $\{x_i\}$ best preserve the pairwise ordering among the entries of $A$, i.e., $\|x_i - x_j\| \leq \|x_k - x_{\ell}\|$ if and only if $a_{ij} \geq a_{k \ell}$; see Chapter 8 of \cite{borg_groenen} for a more detailed overview of non-metric embedding. Given the two collection of graphs, the test procedure in \cite{xixihu} first embed the sample means for each collection using non-metric multidimensional scaling. This yields two $n \times d$ matrices $\hat{X}$ and $\hat{Y}$. The test statistic is given by the Procrustes error
$T(\hat{X}, \hat{Y}) = \min\|\hat{X} - s \hat{Y} W - 1 t^{\top}\|_{F}$
where the minimum is over all scalar $s \in \mathbb{R}$, orthogonal matrix $W \in \mathbb{R}^{d \times d}$ and vector $t \in \mathbb{R}^{d}$. 

Table~\ref{sim4_t1} compares the finite-sample performance of the two test procedures for graphs generated using the same settings as that of Tables~\ref{tbl:permutation}. Table~\ref{sim4_t1} indicates that our test procedure is substantially more powerful than the non-metric embedding test procedure, e.g., compare the power of the two procedures for $\epsilon \leq 0.2$. 
\begin{table}
\centering
\def~{\hphantom{0}}
\caption{Power of the proposed test and non-metric multidimensional scaling  ($\alpha=0.05$).}{
\begin{tabular}{cccccccc}
  $n$&Method&$\epsilon=0$&$\epsilon=0.02$&$\epsilon=0.1$&$\epsilon=0.2$&$\epsilon=0.5$&$\epsilon=1$\\
  \multirow{2}{*}{20}&Proposed test&0.05&0.09&0.44&0.98&1&1\\
  &Non-metric embedding&0.05&0.09&0.11&0.70&0.61&0.86\\
  \multirow{2}{*}{50}&Proposed test&0.05&0.92&1&1&1&1\\
  &Non-metric embedding&0.05&0.24&0.35&0.50&0.98&1\\
  \multirow{2}{*}{100}&Proposed test&0.05&1&1&1&1&1\\
  &Non-metric embedding&0.05&0.10&0.19&0.28&0.70&1\\
  \multirow{2}{*}{200}&Proposed test&0.05&1&1&1&1&1\\
  &Non-metric embedding&0.05&0.06&0.11&0.14&0.51&1\\
  \end{tabular}
  }
\label{sim4_t1}
\end{table}

\section{Empirical Studies}
\subsection{Application 1: connectome data across life span} In this
application, we are interested in determining whether the structural
brain networks of healthy individuals change across their life
span. We used a dataset from \cite{faskowitz2018weighted} where each
network represents connections between $131$ brain regions of interest
and the edges are constructed based on the number of streamlines
connecting these regions. There are in total $622$ networks. The age
for each of the $622$ subjects ranges from $7$ to $85$ years old.

To conduct two-sample comparison, we divide the sample into 3
subgroups according to the subjects' ages, i.e., a young-adult group with ages
in $[18,35]$, a middle-aged group with ages in $(35,56]$ and
an old-adult group with ages in $(56,85]$. The sample sizes for each subgroup
are $171$, $173$ and $207$ graphs, respectively. The number of
vertices in each graph is $131$ and the average edge
densities for the middle-age and old-adult groups are $0.89$ and
$0.95$ that of the young-adult group, respectively.

We then construct pairwise two-sample comparisons between these three
age groups. The associated $p$-values and empirical distributions of
the test statistics for the various null hypotheses are obtained by
permutation test and are illustrated in Figure~\ref{fig:age}. We apply universal
singular value thresholding with dimension $K=3$ chosen according to a
dimension selection algorithm in \cite{zhu_ghodsi}.

\begin{figure}[H]
\centering
\begin{subfigure}{0.32\textwidth}
\includegraphics[width=\textwidth,height=1.375in]{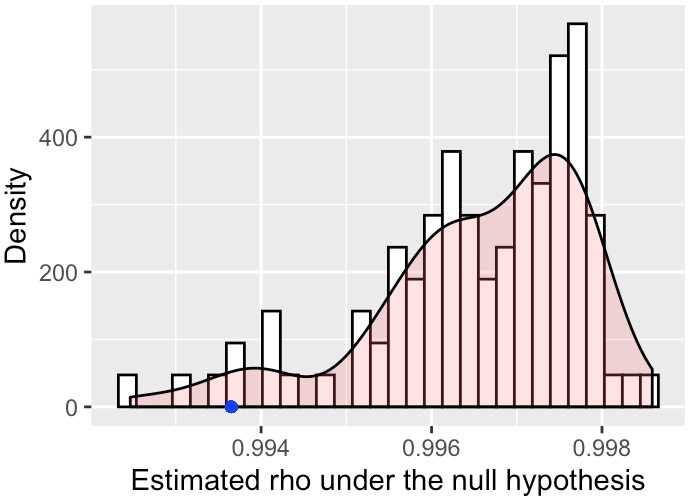}
\caption{$H_0:\textbf{X}_{\text{Young}}=\textbf{X}_{\text{Mid}}$.\\ $T_n=0.9937$ and $p$-value is 0.04.}
\end{subfigure}
\begin{subfigure}{0.32\textwidth}
\includegraphics[width=\textwidth,height=1.375in]{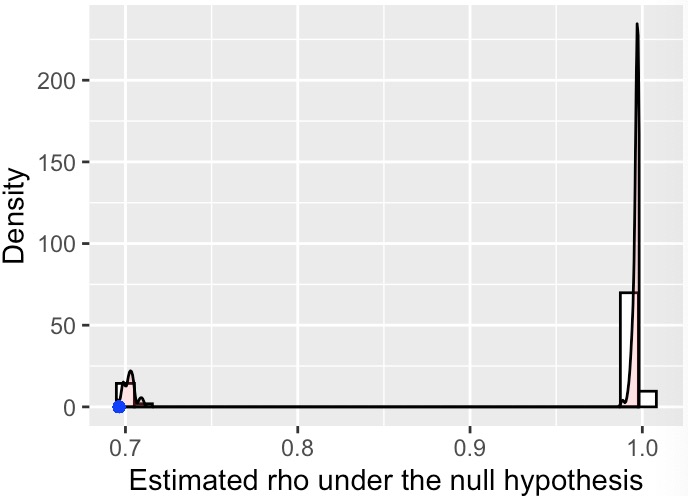}
\caption{$H_0:\textbf{X}_{\text{Mid}}=\textbf{X}_{\text{Old}}$.\\ $T_n=0.6963$ and $p$-value is 0.01.}
\end{subfigure}
\begin{subfigure}{0.32\textwidth}
\includegraphics[width=\textwidth,height=1.38in]{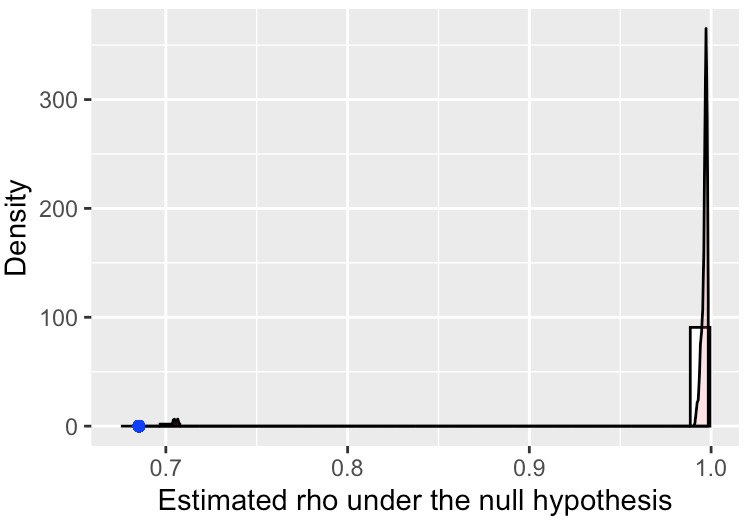}
\caption{$H_0:\textbf{X}_{\text{Young}}=\textbf{X}_{\text{Old}}$.\\ $T_n=0.6851$ and $p$-value is 0.}
\end{subfigure}
\caption{Density Plot of Estimated Test Statistic Under Different $H_0$.}
\label{fig:age}
\end{figure}

The $p$-values given in Figure~\ref{fig:age} are marginal $p$-values
and had not been corrected for multiple comparisons. Applying
Bonferroni correction with significance level $0.05/3\approx0.017$, we
fail to reject the null hypothesis that there is no difference between
the young and the middle-aged group; the value of the test statistic
for this comparison is $T \approx 0.994$. In contrasts, we reject the
null hypothesis for the comparison of young against old and the
comparison of middle-aged against old.  The histograms in Figure~1
also indicate that the empirical distribution of the permutation test
statistic for comparing young and the middle-aged group is tightly
concentrated at $1$, once again suggesting that these two groups are
quite similar.

\subsection{Application 2: epileptogenic data on recording region and
brain state} The second application is on networks constructed from
epileptogenic recordings of patients with epileptic seizure
\cite{andrzejak2001indications}. The data is available from UCI
Machine Learning Repository
(http://archive.ics.uci.edu/ml/datasets.php). There are $500$ subjects
whose brain activity was recorded, with the epileptogenic recording of
each person being divided into $23$ one-second snapshots containing
$178$ time points. The data is arranged as a matrix with $23\times
500=11500$ rows and $178$ columns. The $11500$ observations are
classified into five classes; these classes are numbered from $1$
through $5$ and correspond to recordings with seizure activity, an
area with tumour, a healthy brain area, subject with eyes open and
subject with eyes closed. It was noted in
\cite{andrzejak2001indications} that all subjects whose recordings
are classified as classes $2$ through $5$ are subjects who did not
have epileptic seizure and that only subjects in class $1$ have
epileptic seizure. Most analysis of this data have thus been binary
classification, namely discriminating class $1$ with epileptic seizure
against the rest.

We constructed networks by thresholding the autocorrelation matrices
of the epileptogenic data using a procedure similar to that in
\cite{ghoshdastidar2018practical}. Each class is randomly divided into
four parts with equal size. We then compute the autocorrelation
matrices for each part and set the diagonal elements to be
0. Unweighted adjacency matrices are then obtained by thresholding the
largest $10\%$ of the correlation entries to $1$ with the remaining
entries being $0$. The above steps result in $20$ adjacency matrices,
with $4$ from each class. Each adjacency matrix corresponds to a graph
on $n=178$ vertices.

We then compare, using our test procedure, the graphs from class $1$
against the graphs from class $j \geq 2$. The results are summarized
in Table~\ref{app2}. The $p$-values in the table are calculated using
permutation test. We see from Table \ref{app2} that class $1$ is
significantly different from the remaining classes and that this
difference is not too sensitive to the choice of dimension $K$ in the
singular value thresholding step.

\begin{table}
\centering
\def~{\hphantom{0}}
\caption{$p$-values for epileptogenic correlation networks. }{%
    \begin{tabular}{cccccc}
         $K$&$A_1$ vs $A_1$&$A_1$ vs $A_2$&$A_1$ vs $A_3$&$A_1$ vs $A_4$&$A_1$ vs $A_5$ \\
         2&0.80&0.99&0.12&0.04&0.01\\
         3&0.80&0.72&0.01&0.00&0.29\\
         4&0.79&\textbf{0.01}&\textbf{0.00}&\textbf{0.04}&\textbf{0.01}\\
         5&0.88&\textbf{0.01}&\textbf{0.00}&\textbf{0.00}&\textbf{0.01}\\
         6&0.80&\textbf{0.04}&\textbf{0.00}&\textbf{0.01}&\textbf{0.02}\\
         7&0.94&\textbf{0.03}&\textbf{0.00}&\textbf{0.01}&\textbf{0.02}\\
         8&0.83&0.05&0.00&0.00&0.03\\
    \end{tabular}}
    \label{app2}
\end{table}


\section{Discussions}
In summary, the test statistic constructed based on universal singular value thresholding and Spearman's rank correlation coefficient yields a valid and consistent test procedure for testing whether two latent distance random graphs on the same vertex set have the similar generating latent positions. A few related questions will be left for future research.

Firstly, for the two-sample hypothesis test we study, one
can also develop test statistics using other techniques, for example, isotonic regression. As we briefly introduced in Section~\ref{sec:methodology}, when the null hypothesis is true then there exists a monotone function $f$ such that $p_{ij}=f(q_{ij})$. Thus, given graphs from the latent distance model, we can first estimate $\hat{P}$ and $\hat{Q}$ and then fit a regression model of the form $\hat{p}_{ij} = f(\hat{q}_{ij}) + \epsilon_{ij}$ for some nonparametric function $f$. The two-sample testing problem can then be reformulated as testing for whether $f$ is monotone. It appears, however, that testing for monotonicity against a general alternative is still an open problem in nonparametric regression.

The second question concerns the rate of convergence of our test statistic and the class of alternatives for which the test procedure is consistent against. In particular Theorem~\ref{thm:consistency} shows that the test procedure is consistent for class of alternatives where the distance between the 
collections of latent positions diverges with rate $\Omega(n^{1/2})$. Relaxing this condition will require careful analysis of the estimation error in the singular value thresholding procedure as well as convergence rate for non-metric embedding.

Finally, the critical region for our test procedures are determined using resampling methods such as permutation test or bootstrapping graphs from the estimated edge probabilities.
The validity of these resampling techniques are justified by the
empirical simulation studies as well as real data
analysis. Nevertheless, our test procedure could be more robust if we
are able to derive the limiting distribution of the test statistic and
thereby obtain approximate critical values. We surmise, however, that
this will be a challenging problem. Indeed, while the limiting
variance and distribution of Spearman’s rank correlation in the case
of independent and identically distributed data are well known, see e.g., \cite{kendall1948rank} and \cite{ruymgaart}, the entries in our estimates $\hat{P}$ and $\hat{Q}$ are not independent, and furthermore the original entries of $P$ and $Q$ are not identically distributed. 

\bibliographystyle{apalike}
\bibliography{manuscript}
\newpage
\section{Supplementary}

This supplementary file contains proofs of the theoretical results in the main paper and an additional simulation experiment.

\subsection{Proofs}
In this section, we will show the detailed proof of Theorem 3.1 and 3.2. Before that, let us recall the necessary notations and assumptions given in Section 3 of the main paper.
\vskip 0.5cm
\noindent\textbf{Notations.} Let $n$ be the number of vertices. $P\in\mathbb{R}^{n\times n}$ is the binary symmetric edge-probability matrix. $R(P)$ is a symmetric matrix measuring ranks corresponding to $P$. Denote universal singular value thresholding  estimated edge-probability matrix as $\hat{P}$. For $\eta>0$ and $(\eta^2\rho)^{-1}=o(n)$, we define the discretization as $\tilde{p}_{ij}=\ceil* {\hat{p}_{ij}/\eta}\times\eta$.
Let $\|\cdot\|_F$ be the Frobenius norm.
\vskip 0.5cm

\noindent\textbf{Assumptions.} We made several required assumptions as follows.
\begin{enumerate}
    \item As $n\to\infty$, for any $\epsilon>0$, there exists $\delta>0$ such that $[0,1)$ can be partitioned into the union of intervals of the form $[(k-1)\delta,k\delta)$ for $k=1,...,\lceil 1/\delta \rceil$, such that, for any $k$, one of the following two conditions holds almost surely:
\begin{enumerate}[label=(\roman*)]
	\item Either the number of $ij$ pairs with $i < j$ and $p_{ij}\in[(k-1)\delta,k\delta)$ is at most $n(n-1)\epsilon/2$.
	\item Or if the number of $ij$ pairs with $i < j$ and $p_{ij}\in[(k-1)\delta,k\delta)$ exceeds $n(n-1)\epsilon/2$, then they are all equal for $p_{ij} \in [(k-1)\delta,k\delta)$.
\end{enumerate}
\item Define $\hat{\sigma}\{R(P)\}=\Bigl[{n\choose2}^{-1}\sum_{i<j}\Bigl\{R(p_{ij})-{n\choose2}^{-1}\sum_{i<j}R(p_{ij})\Bigr\}^2\Bigr]^{1/2}$ and $\hat{\sigma}\{R(Q)\}$ similarly. It holds that 
$\hat{\sigma}\{R(P)\}=\Omega(n^2)$ and $\hat{\sigma}\{R(Q)\}=\Omega(n^2)$.
\item The link functions $h$ and $g$ are infinitely many times differentiable.
\item Every edge is observed independently with probability $\rho\in(0,1]$, where there exists a positive constant $C$ such that $n\rho\geq C\log n$.
\item There exists a constant $c>0$ independent of $\delta$ such that $|\{(i,j):p_{ij}\in [(k-1)\delta,k\delta]\}|\leq c\cdot \delta{\tbinom{n}{2}}$.
\item Let $U\subset\mathbb{R}^d$ and $V \subset \mathbb{R}^{d}$, be
  bounded and connected sets. Let $\Omega_n=\{{x}_1,...,{x}_n\}\subset U$ and
  $\Xi_{n} = \{y_1, y_2, \dots, y_n\} \subset V$. Then $\lim_{n
    \rightarrow \infty} \Omega_n$
  is dense in $U$ and $\lim_{n \rightarrow \infty} \Xi_n$ is
  dense in $V$. Furthermore, for any $\epsilon > 0$
  there exists $\delta_U = \delta_U(\epsilon) > 0$ and $\delta_V = \delta_V(\epsilon) > 0$ such that
  \begin{gather*}
  n^{-1} \liminf |B(x,\epsilon) \cap \Omega_n| \geq \delta_U, \qquad
  \text{for all $x \in U$}. \\
  n^{-1}\liminf |B(y,\epsilon) \cap \Xi_n| \geq \delta_V, \qquad
  \text{for all $y \in V$}.
  \end{gather*}
  Here $B(x, \epsilon)$ denote the ball of radius $\epsilon$ around $x
  \in \mathbb{R}^{d}$. 
 \end{enumerate}

We now prove Theorem 3.1.
\vskip 0.5cm
\noindent\textsc{Theorem 3.1} 
Assume Assumptions 1--4 hold. Then for sufficiently large $n$,
$$T_n(\tilde{P},\tilde{Q})-T_n(P,Q)=o_p(1).$$
Here $\tilde{P}$ and $\tilde{Q}$ are the $\eta$-discretization of $\hat{P}$ and $\hat{Q}$ with $(\eta^2 \rho)^{-1} = o(n)$ as $n \rightarrow \infty$. 

\begin{proof}[of Theorem 3.1]
From Theorem 1 in \cite{xu2017rates}, along with the conditions in
Assumptions 3 and 4, we have
\begin{align}
\label{universal singular value thresholding }
  \|\hat{P}-P\|^2_F=O_p\Bigl\{\frac{n\log^d(n\rho)}{\rho}\Bigr\},
\end{align} where $d$ is the dimension of latent positions. Let $\eta
> 0$ be such that $(\eta^2 \rho)^{-1} = o(n)$. Define
$S=\Bigl\{(i,j):|\hat{p}_{ij}-p_{ij}|>\eta\Bigr\}$. 
Then by (\ref{universal singular value thresholding }), we have
$$|S|=O\Bigl[n\Bigl\{\log^d(n\rho)\Bigr\}/(\eta^2\rho)\Bigr]=o(n^2).$$
Recall that our test statistic, using the true $P$ and $Q$, is
\begin{align*}
T_n(P,Q)&=\frac{\mathrm{cov}\{R(P),R(Q)\}}{\hat{\sigma}\{R(P)\}\hat{\sigma}\{R(Q)\}},
\end{align*}
where\begin{align*}
 	\mathrm{cov}\{R(P),R(Q)\}&=\tbinom{n}{2}^{-1}\sum_{i,j}R(p_{ij})R(q_{ij})-\Bigl\{\tbinom{n}{2}^{-1}\sum_{i,j}R(p_{ij})\Bigr\}\times \Bigl\{\tbinom{n}{2}^{-1}\sum_{i,j}R(q_{ij})\Bigr\},\\
 	\hat{\sigma}\{R(P)\}&=\Bigl[\tbinom{n}{2}^{-1}\sum_{i,j}\Bigl\{R(p_{ij})-\tbinom{n}{2}^{-1}\sum_{i,j}R(p_{ij})\Bigr\}^2\Bigr]^{1/2},\\
 	\hat{\sigma}\{R(Q)\}&=\Bigl[\tbinom{n}{2}^{-1}\sum_{i,j}\Bigl
\{R(q_{ij})-\tbinom{n}{2}^{-1}\sum_{i,j}R(q_{ij})\Bigr\}^2\Bigr]^{1/2}.
 \end{align*}
The test statistic $T_n(\tilde{P}, \tilde{Q})$ using the discretized
estimates is defined analogously. We then have
\begin{equation*}
  \begin{split}
T_n(\tilde{P},\tilde{Q})-T_n(P,Q)
&=\frac{\mathrm{cov}\{R(\tilde{P}),R(\tilde{Q})\}}{\hat{\sigma}\{R(\tilde{P})\}\hat{\sigma}\{R(\tilde{Q})\}}-\frac{\mathrm{cov}\{R(P),R(Q)\}}{\hat{\sigma}\{R(P)\}\hat{\sigma}\{R(Q)\}}\\
	&=\underbrace{\frac{\mathrm{cov}\{R(\tilde{P}),R(\tilde{Q})\}-\mathrm{cov}\{R(P),R(Q)\}}{\hat{\sigma}\{R(\tilde{P})\}\hat{\sigma}\{R(\tilde{Q})\}}}_{\text{Part I}}\\
	&+\underbrace{\mathrm{cov}\{R(P),R(Q)\}\Bigl[\frac{1}{\hat{\sigma}\{R(\tilde{P})\}\hat{\sigma}\{R(\tilde{Q})\}}-\frac{1}{\hat{\sigma}\{R(P)\}\hat{\sigma}\{R(Q)\}}\Bigr]}_{\text{Part II}}.
\end{split}
\end{equation*}
We control Part I and Part II via the following two lemmas. 
\begin{lemma} Under assumptions 1,3 and 4, for any $\epsilon>0$, there exists a positive constant $C$ such that
$$\Bigl|\sum_{i<j}R(\tilde{p}_{ij})\sum_{i<j}R(\tilde{q}_{ij})-\sum_{i<j}R(p_{ij})\sum_{i<j}R(q_{ij})\Bigr|\leq Cn^8\epsilon.$$
\end{lemma}
\begin{lemma} Under assumptions 1,3 and 4, for any $\epsilon>0$, there exists a positive constant $C$ such that
$$\Bigl|\hat{\sigma}^2\{R(\tilde{P})\}-\hat{\sigma}^2\{R(P)\}\Bigr|\leq C n^4\epsilon.$$
\end{lemma}
Suppose Lemma 1 and 2 are valid then we can complete the proof of Theorem 3.1. Lemma 1 is used to control the numerator in each part while Lemma 2 is for the denominator. For any $\epsilon>0$, we have $\mathrm{cov}\{R(\tilde{P}),R(\tilde{Q})\}-\mathrm{cov}\{R(P),R(Q)\}=O_p(n^4\epsilon)$ by Lemma 1 and $\hat{\sigma}\{R(\tilde{P})\}=\hat{\sigma}\{R(P)\}+O_p(n^2\epsilon^{1/2})$ by Lemma 2. Under Assumption 2, $\hat{\sigma}\{R(P)\}=\Omega(n^2)$. Thus, it holds that
$\hat{\sigma}\{R(\tilde{P})\}=\Omega_p(n^2)$.
Similarly, $\hat{\sigma}\{R(\tilde{Q})\}=\Omega_p(n^2)$. Therefore, we have
$$\text{Part I}=\frac{\mathrm{cov}\{R(\tilde{P}),R(\tilde{Q})\}-\mathrm{cov}\{R(P),R(Q)\}}{\hat{\sigma}\{R(\tilde{P})\}\hat{\sigma}\{R(\tilde{Q})\}}=O_p({\epsilon}^{1/2}).$$
Then,
\begin{align*}
\text{Part II} = & \mathrm{cov}\{R(P),R(Q)\}\Bigl[\frac{1}{\hat{\sigma}\{R(\tilde{P})\}\hat{\sigma}\{R(\tilde{Q})\}}-\frac{1}{\hat{\sigma}\{R(P)\}\hat{\sigma}\{R(Q)\}}\Bigr]\\
\leq&\hat{\sigma}\{R(P)\}\hat{\sigma}\{R(Q)\}\frac{\hat{\sigma}\{R(P)\}\hat{\sigma}\{R(Q)\}-\hat{\sigma}\{R(\tilde{P})\}\hat{\sigma}\{R(\tilde{Q})\}}{\hat{\sigma}\{R(\tilde{P})\}\hat{\sigma}\{R(\tilde{Q})\}\hat{\sigma}\{R(P)\}\hat{\sigma}\{R(Q)\}}\\
=&\frac{\hat{\sigma}\{R(P)\}\hat{\sigma}\{R(Q)\}-\hat{\sigma}\{R(\tilde{P})\}\hat{\sigma}\{R(\tilde{Q})\}}{\hat{\sigma}\{R(\tilde{P})\}\hat{\sigma}\{R(\tilde{Q})\}}\\
=&\frac{\Bigl[\hat{\sigma}\{R(P)\}-\hat{\sigma}\{R(\tilde{P})\}\Bigr]\hat{\sigma}\Bigl\{R(Q)\Bigr]+\hat{\sigma}\{R(\tilde{P})\}\Bigl[\hat{\sigma}\{R(Q)\}-\hat{\sigma}\{R(\tilde{Q})\}\Bigr]}{\Omega_p(n^4)}\\
=&\frac{O_p(n^2\epsilon^{1/2})O(n^2)+\Bigl\{O_p(n^2\epsilon^{1/2})+O(n^2)\Bigr\}O_p(n^2\epsilon^{1/2})}{\Omega_p(n^4)}\\
=&\frac{2O_p(n^2\epsilon^{1/2})O(n^2)+O_p(n^4\epsilon)}{\Omega_p(n^4)}
=O_p({\epsilon}^{1/2}).
\end{align*}
Combining the above bounds yields
$$T_n(\tilde{P},\tilde{Q})-T_n(P,Q)=O_p({\epsilon}^{1/2}).$$
Since $\epsilon > 0$ is arbitrary, we have
$T_n(\tilde{P},\tilde{Q})-T_n(P,Q)=o_p(1)$ as desired.
\end{proof}
We now prove Lemmas 1 and 2. The proof of Lemma 1
depends on the following result.
\begin{lemma} Under assumptions 1,3 and 4, for any $\epsilon>0$, there exists a positive constant $C$ such that
$$\Bigl|\sum_{i<j}\Bigl\{R(\tilde{p}_{ij})R(\tilde{q}_{ij})-R(p_{ij})R(q_{ij})\Bigr\}\Bigr|\leq Cn^6\epsilon.$$
\end{lemma}
\begin{proof}[of Lemma 3]
  Our proof is based on bounding $|R(\tilde{p}_{rs}) - R(p_{rs})|$ and
  $|R(\tilde{q}_{rs}) - R(q_{rs})|$. First consider pairs $(i,j)\notin S$. Since
  $|\hat{p}_{ij}-p_{ij}|\leq\eta$ and
  $|\tilde{p}_{ij}-\hat{p}_{ij}|\leq\eta$, we have
  $|\tilde{p}_{ij}-p_{ij}| \leq 2\eta$. Now suppose that
  $\tilde{p}_{ij}\leq x$. Then $p_{ij}\leq x+2\eta$.  We therefore have
  \begin{gather*}
    |\{(i,j):\tilde{p}_{ij}\leq x\}| \leq
    |\{(i,j):p_{ij}\leq x+2\eta\}|, \quad |\{(i,j):\tilde{p}_{ij}<x\}|\geq|\{(i,j):p_{ij}<
      x-2\eta\}|.
  \end{gather*}
Thus, for a given pair $(r,s)\notin \mathcal{S}$, we have
\begin{align*}
R(\tilde{p}_{rs})&\leq|\{(i,j):\tilde{p}_{ij}\leq \tilde{p}_{rs}\}|\leq |\{(i,j):p_{ij}\leq \tilde{p}_{rs}+2\eta\}|\leq |\{(i,j):p_{ij}\leq p_{rs}+4\eta\}|\\
&\leq R(p_{rs})+5 \tbinom{n}{2} \epsilon.
\end{align*}
A similar argument shows, for $(r,s) \not \in\mathcal{S}$, 
\begin{align*}
R(\tilde{p}_{rs})&\geq |\{(i,j):\tilde{p}_{ij}< \tilde{p}_{rs}\}|\geq |\{(i,j):p_{ij}<\tilde{p}_{rs}-2\eta\}|\geq |\{(i,j):p_{ij}< p_{rs}-4\eta\}|\\
&=|\{(i,j):p_{ij}\leq p_{rs}-4\eta\}|-|\{(i,j):p_{ij}= p_{rs}-4\eta\}|\\
&\geq R(p_{rs})- 4 \tbinom{n}{2} \epsilon- \tbinom{n}{2} 
\epsilon =R(p_{rs})-5 \tbinom{n}{2} \epsilon.
\end{align*}
Therefore, for $(r,s)\notin S$ and any $\epsilon >0$, we have
\begin{align}
\label{lm1eq1}
|R(\tilde{p}_{rs})-R(p_{rs})|\leq5 \tbinom{n}{2} \epsilon.
\end{align}
A similar argument yields
\begin{align}
\label{lm1eq2}
|R(\tilde{q}_{rs})-R(q_{rs})|\leq5 \tbinom{n}{2} \epsilon.
\end{align}
By applying (\ref{lm1eq1}) and  (\ref{lm1eq2}), for $(i,j)\notin S$ and any $\epsilon >0$, it holds that 
\begin{align*}
&\Bigl|R(\tilde{p}_{ij})R(\tilde{q}_{ij})-R(p_{ij})R(q_{ij})\Bigr|\\
& \leq\Bigl|R(\tilde{p}_{ij})-R(p_{ij})\Bigr|R(q_{ij})+\Bigl|R(\tilde{q}_{ij})-R(q_{ij})\Bigr|R(p_{ij})+\Bigl|R(\tilde{p}_{ij})-R(p_{ij})\Bigr|\Bigl|R(\tilde{q}_{ij})-R(q_{ij})\Bigr|\\
  &\leq5 \tbinom{n}{2} \epsilon \times R(q_{ij})+ 5 \tbinom{n}{2}
    \epsilon \times R(p_{ij})+ 25 \tbinom{n}{2}^2\epsilon^2. 
\end{align*}
Now consider $(i,j)\in S$. Then by assumption 3 and 4, $$\sum_{(i,j)\in S}\Bigl|R(\tilde{p}_{ij})R(\tilde{q}_{ij})-R(p_{ij})R(q_{ij})\Bigr|\leq O_p(|S|\cdot n^4)=o_p(n^6).$$ Therefore, for any $1\leq i<j\leq n$, there exists a positive constant $C$ such that
\begin{align*}
 &\Bigl|\sum_{1\leq i<j\leq n}R(\tilde{p}_{ij})R(\tilde{q}_{ij})-R(p_{ij})R(q_{ij})\Bigr|\\
 &\leq\sum_{(i,j)\notin S}\Bigl|R(\tilde{p}_{ij})R(\tilde{q}_{ij})-R(p_{ij})R(q_{ij})\Bigr|+\sum_{(i,j)\in S}\Bigl|R(\tilde{p}_{ij})R(\tilde{q}_{ij})-R(p_{ij})R(q_{ij})\Bigr|\\
 &\leq5 \tbinom{n}{2} \epsilon \sum_{(i,j)\notin S}R(p_{ij})+ 5
   \tbinom{n}{2} \epsilon \sum_{(i,j)\notin S}R(q_{ij})+25
   \tbinom{n}{2}^3 \epsilon^2+o_p(n^6) \\ 
 &\leq Cn^6\epsilon
\end{align*}
as desired.
\end{proof}
\begin{proof}[of Lemma 1]
By (\ref{lm1eq1}) and (\ref{lm1eq2}),
\begin{align*}
&\Bigl|\sum_{(i,j)\notin S}R(\tilde{p}_{ij})\sum_{(i,j)\notin S}R(\tilde{q}_{ij})-\sum_{(i,j)\notin S}R(p_{ij})\sum_{(i,j)\notin S}R(q_{ij})\Bigr|\\
\leq& \Bigl|\sum_{(i,j)\notin S}R(\tilde{p}_{ij})-\sum_{(i,j)\notin S}R(p_{ij})\Bigr|\sum_{(i,j)\notin S}R(q_{ij})+\Bigl|\sum_{(i,j)\notin S}R(\tilde{q}_{ij})-\sum_{(i,j)\notin S}R(q_{ij})\Bigr|\sum_{(i,j)\notin S}R(p_{ij})\\
&+\Bigl|\sum_{(i,j)\notin S}R(\tilde{p}_{ij})-\sum_{(i,j)\notin S}R(p_{ij})\Bigr|\Bigl|\sum_{(i,j)\notin S}R(\tilde{q}_{ij})-\sum_{(i,j)\notin S}R(q_{ij})\Bigr|\\
\leq&\sum_{(i,j)\notin S}\Bigl|R(\tilde{p}_{ij})-R(p_{ij})\Bigr|\sum_{(i,j)\notin S}R(q_{ij})+\sum_{(i,j)\notin S}\Bigl|R(\tilde{q}_{ij})-R(q_{ij})\Bigr|\sum_{(i,j)\notin S}R(p_{ij})\\
&+\sum_{(i,j)\notin S}\Bigl|R(\tilde{p}_{ij})-R(p_{ij})\Bigr|\sum_{(i,j)\notin S}\Bigl|R(\tilde{q}_{ij})-R(q_{ij})\Bigr|\\
=&10{\tbinom{n}{2}}^2\epsilon \cdot O(n^4)+25{\tbinom{n}{2}}^4\epsilon^2\leq Cn^8\epsilon.
\end{align*}
For $(i,j)\in S$, under assumption 3 and 4, it holds that 
$$\Bigl|\sum_{(i,j)\in S}R(\tilde{p}_{ij})\sum_{(i,j)\in S}R(\tilde{q}_{ij})-\sum_{(i,j)\in S}R(p_{ij})\sum_{(i,j)\in S}R(q_{ij})\Bigr|\leq O_p(|S|^2n^4)=o_p(n^8).$$
Therefore, for $1\leq i<j\leq n$ and any $\epsilon >0$, there exists a positive constant $C$ such that
\begin{align*}
\Bigl|\sum_{i<j}R(\tilde{p}_{ij})\sum_{i<j}R(\tilde{q}_{ij})-\sum_{i<j}R(p_{ij})\sum_{i<j}R(q_{ij})\Bigr|
\leq  Cn^8\epsilon.
\end{align*}
\end{proof}

\begin{proof}[of Lemma 2]
By (\ref{lm1eq1}), for $(i,j)\notin S$ and any $\epsilon>0$,  we have
\begin{align}
\label{lm2eq1}
	\sum_{(i,j)\notin S}\Bigl|R^2(\tilde{p}_{ij})-R^2(p_{ij})\Bigr|&\leq2\sum_{(i,j)\notin S}\Bigl|R(\tilde{p}_{ij})-R(p_{ij})\Bigr|R(p_{ij})+\sum_{(i,j)\notin S}\Bigl|R(\tilde{p}_{ij})-R(p_{ij})\Bigr|^2\nonumber\\
	&\leq10{\tbinom{n}{2}}\epsilon\cdot \sum_{(i,j)\notin S}R(p_{ij})+25{\tbinom{n}{2}}^3\epsilon^2\nonumber\\
	&=10{\tbinom{n}{2}}\epsilon\cdot  O(n^4)+25{\tbinom{n}{2}}^3\epsilon^2=Cn^6\epsilon.
\end{align}
Consider $(i,j)\in S$. Similarly, by assumption 3 and 4, it holds that
\begin{align}
\label{lm2eq2}
	\sum_{(i,j)\in S}\Bigl|R^2(\tilde{p}_{ij})-R^2(p_{ij})\Bigr|&\leq2\sum_{(i,j)\in S}\Bigl|R(\tilde{p}_{ij})-R(p_{ij})\Bigr|R(p_{ij})+\sum_{(i,j)\in S}\Bigl|R(\tilde{p}_{ij})-R(p_{ij})\Bigr|^2\nonumber\\
	&=O_p(|S|n^4)=o_p(n^6).
\end{align}

Thus, combining (\ref{lm2eq1}) and (\ref{lm2eq2}), for $1\leq i<j\leq n$ and any $\epsilon >0$, there exists $C>0$ such that 
\begin{align}
\label{lm2eq3}
\sum_{i<j}\Bigl|R^2(\tilde{p}_{ij})-R^2(p_{ij})\Bigr|\leq C n^6\epsilon.	
\end{align}
Following the similar procedure of deriving (\ref{lm2eq3}), it is easy to show for any $\epsilon >0$, there exists $C>0$ such that 
\begin{align}
\label{lm2eq4}
	\Bigl|\Bigl\{{\tbinom{n}{2}}^{-1}\sum_{i<j}R(\tilde{p}_{ij})\Bigr\}^2-\Bigl\{{\tbinom{n}{2}}^{-1}\sum_{i<j}R(p_{ij})\Bigr\}^2\Bigr|\leq Cn^4\epsilon.
\end{align}
By (\ref{lm2eq3}) and (\ref{lm2eq4}), for any $\epsilon >0$, there exists a positive constant $C$ such that
\begin{align*}
\small
&\Bigl|{\tbinom{n}{2}}^{-1}\sum_{i<j}\Bigl\{R(\tilde{p}_{ij})-\sum_{i<j}R(\tilde{p}_{ij})\Bigr\}^2-{\tbinom{n}{2}}^{-1}\sum_{i<j}\Bigl\{R(p_{ij})-{\tbinom{n}{2}}^{-1}\sum_{i<j}R(p_{ij})\Bigr\}^2\Bigr|\\
	\leq & {\tbinom{n}{2}}^{-1}\sum_{i<j}\Bigl|R^2(\tilde{p}_{ij})-R^2(p_{ij})\Bigr|+\Bigl|\Bigl\{{\tbinom{n}{2}}^{-1}\sum_{i<j}R(\tilde{p}_{ij})\Bigr\}^2-\Bigl\{{\tbinom{n}{2}}^{-1}\sum_{i<j}R(p_{ij})\Bigr\}^2\Bigr|\\
	\leq& Cn^4\epsilon.
\end{align*}
\end{proof}

Then we will prove Theorem 3.2, which requires the following lemma.

\begin{lemma}
Given $X=({x}_1,...,{x}_n)^{\top}\in\mathbb{R}^{n\times d}$ and
$Y=({y}_1,...,{y}_n)^{\top}\in\mathbb{R}^{n\times d}$. Define
$P=(p_{ij})\in\mathbb{R}^{n\times n}$ and
$Q=(q_{ij})\in\mathbb{R}^{n\times n}$, where
$p_{ij}=h(\|{x}_i-{x}_j\|)$ and $q_{ij}=g(\|{y}_i-{y}_j\|)$ for some
monotone decreasing functions $h,g$ from $\mathbb{R}$ onto
$\mathbb{R}$.
Under Assumption 3 and 6, if $\lim_{n\to\infty}T_n(P,Q)=1$ then there
exists a sequence of
monotone increasing functions $f_n$ from $\mathbb{R}$ onto $\mathbb{R}$ such that, as $n\to\infty$,
	$$\max_{i,j}\Bigl|\|{y}_i-{y}_j\|-f_n\Bigl(\|{x}_i-{x}_j\|\Bigr)\Bigr|
    \rightarrow 0.$$
\end{lemma}

\begin{proof}[of Lemma 4]
Since $T_n(P, Q) \rightarrow 1$, for any $\epsilon > 0$ there exists a
universal constant $C$ and a $n_0 = n_0(\epsilon)$ such that if $n \geq n_0$ then the number of pairs
$\{i,j\}$ with $|R(p_{ij}) - R(q_{ij})| \geq \tbinom{n}{2} \epsilon$ is at most $C\tbinom{n}{2} \epsilon$. Let $\mathcal{S}$ be the set of
pairs satisfying 
\begin{equation}
  \label{eq2_2}
  |R(p_{ij}) - R(q_{ij})| < \tbinom{n}{2} \epsilon.
\end{equation}
Define rank functions normalized by ${n \choose 2}$ as
$\tilde{R}_p,\tilde{R}_q:[0,1]\mapsto[0,1]$. We can rewrite (\ref{eq2_2}) as
$$\Bigl|\tilde{R}_q\Bigl\{g(\|{y}_i - {y}_j\|)\Bigr\}-\tilde{R}_p\Bigl\{h(\|{x}_i - {x}_j\|)\Bigr\}\Bigr|<C\epsilon.$$
According to Assumption 3, $ \tilde{R}_q\circ g$
is uniformly continuous. Thus, for all $\{i,j\} \in \mathcal{S}$,
\begin{align*}
    \|{y}_i - {y}_j\|&\in\Bigl[g^{-1}\circ
                       \tilde{R}_q^{-1}\Bigl\{\tilde{R}_p\circ
                       h(\|{x}_i - {x}_j\|)\pm C\epsilon\Bigr\}\Bigr]\\
    &\in\Bigl[g^{-1}\circ \tilde{R}_q^{-1}\Bigl\{\tilde{R}_p\circ h(\|{x}_i - {x}_j\|)\Bigr\}\pm \epsilon_1\Bigr],
\end{align*}
where $\epsilon_1>0$ depends on $g^{-1}\circ \tilde{R}_q^{-1}$, $C$
and $\epsilon$. Also, $\epsilon_1 \rightarrow 0$ as $\epsilon
\rightarrow 0$ due to the uniform continuity of $\tilde{R}_q \circ g$.
Define a sequence of functions $f_n=g^{-1}\circ
\tilde{R}_q^{-1}\circ \tilde{R}_p\circ h$. We then have, for all
$\{i,j\} \in \mathcal{S}$,
$$\max_{i,j}\Bigl|\|{y}_i-{y}_j\|-f_n\Bigl(\|{x}_i-{x}_j\|\Bigr)\Bigr|\to 0,$$
as $n\to\infty$.

We next consider the pairs $\{i,j\} \not \in \mathcal{S}$. Suppose
first that there exists a pair $\{k, \ell\} \in \mathcal{S}$ such that
both $R(p_{ij}) - R(p_{k \ell}) \leq \tbinom{n}{2} \epsilon$ and
$R(q_{ij}) - R(q_{k \ell}) \leq \tbinom{n}{2} \epsilon$. Then by the
continuity and monotonicity of $f_n$, we have
$$\|y_{i} - y_{j}\| \in \Bigl(\|y_{k} - y_{\ell}\| \pm C
\epsilon_1\Bigr) \subset  \Bigl( f(\|x_{k} - x_{\ell}\| \pm \epsilon)
\pm C \epsilon_1\Bigr) 
$$
and by taking $\epsilon$ (and hence $\epsilon_1$) sufficiently small,
we have
$$\|y_{i} - y_{j}\| - f_n(\|x_i - x_j\|) \rightarrow 0$$
as $n \rightarrow \infty$.

It remains to consider the pairs $\{i,j\} \not \in \mathcal{S}$ such
that either $|R(p_{ij}) - R(p_{k \ell})| \geq \tbinom{n}{2} \epsilon$
for all $\{k, \ell\} \in \mathcal{S}$ or that $|R(q_{ij}) - R(q_{k
  \ell})| \geq \tbinom{n}{2} \epsilon$ for all $\{k, \ell\} \in
\mathcal{S}$. Suppose $|R(p_{ij}) - R(p_{k \ell})| \geq \tbinom{n}{2}
\epsilon$ for all $\{k,\ell\} \in \mathcal{S}$. Then there exists a
$\delta > 0$ such that for all $i', j'$, if
$$\|x_{i'} - x_{i}\| \leq \delta, \quad \|x_{j'} - x_j\| \leq \delta$$
then $|R(p_{i'j'}) - R(p_{ij})| \leq \tbinom{n}{2} \epsilon$, i.e., we
have $\{i',j'\} \not \in \mathcal{S}$. That is
to say, points $x_{i'}$ ``close'' to $x_i$ and $x_{j'}$ ``close``
to $x_j$ will have distance $\|x_{i'} - x_{j'}\|$ ``close'' to $\|x_i
- x_j\|$ and hence the ranks of $p_{i'j'}$ and $p_{ij}$ are
``close''. Assumption~6 then implies that number of points in $B(x_i, \delta)$ is of order $\Omega(n)$ as $n \rightarrow \infty$ and since
$|\mathcal{S}^{c}|$ has at most $C\tbinom{n}{2} \epsilon$ elements,
the vertex covering number for $\mathcal{S}^{c}$ is of order $O(n \epsilon^{1/2})$ as
$n \rightarrow \infty$. We can then remove these vertices from
consideration. We repeat the same procedure for $Q$.

In summary, if $T_n(P, Q) \rightarrow 1$ then there is a subset of
$\mathcal{T}$ rows of both $X$ and $Y$ such that $|\mathcal{T}| = n -
O(n \epsilon^{1/2})$ and $$\|y_{i} - y_{j}\| - f_n(\|x_i - x_j\|) \rightarrow 0, \quad i,j
\in \mathcal{T}.$$
As $\epsilon > 0$ is arbitrary, we can have $|\mathcal{T}| = n - o(n)$
for sufficiently large $n$.

Once again, by Assumption $6$, any sequence of $n - o(n)$ elements $\{x_i
\colon i \in \mathcal{T}\}$ will be dense in $U$ as $n \rightarrow \infty$, and the corresponding
$\{y_i \colon i \in \mathcal{T}\}$ will be dense in $V$. We therefore have,
$$\max_{ij} \Bigl|\|y_{i} - y_{j}\| - f_n(\|x_i - x_j\|)\Bigr| \longrightarrow 0$$
as $n \rightarrow \infty$ as desired.

\end{proof}

\noindent\textbf{Theorem 3.2} Under Assumption 3 and 6, if $T_n(P,Q)\to 1$ as $n\to\infty$, it holds that there exists $s
\in\mathbb{R}$, orthogonal $W\in\mathbb{R}^{d\times d}$ and ${t}\in\mathbb{R}^d$ such that 
$$\|X-sYW-1t^{\top}\|_F=o(n^{1/2}).$$

\begin{proof}[of Theorem 3.2]
Let $\phi_n:\Xi_n\mapsto \Omega_n\subset\mathbb{R}^d$ be a function with values in a bounded set $\Omega_n=\{{x}_1,...,{x}_n\}$.
Let $\Xi = \lim_{n \rightarrow \infty} \Xi_n$.

Now take ${y}_i,{y}_j,{y}_k\in\Xi$ such that
$\|{y}_i-{y}_j\|<\|{y}_i-{y}_k\|$. By definition, there is $m$ such
that ${y}_i,{y}_j,{y}_k\in\Xi_m$. Therefore, by Lemma 4, for any $0<
\epsilon \leq (\|{y}_i-{y}_k\|-\|{y}_i-{y}_j\|)/2$, there exists a
$m'\geq m$ such that $f_n(\|\phi_n({y}_i)-\phi_n({y}_j)\|)\leq
f_n(\|\phi_n({y}_i)-\phi_n({y}_k)\|)$ for any $n\geq m'$. Since $f_n
\colon \mathbb{R}\mapsto\mathbb{R}$ is an increasing function, we have $\|\phi_n({y}_i)-\phi_n({y}_j)\|\leq \|\phi_n({y}_i)-\phi_n({y}_k)\|$ for any $n\geq m'$.

By Lemma 2 in \cite{arias2017some}, since $\Xi_n\subset\mathbb{R}^d$ is finite and $\phi_n:\Xi_n\mapsto \Omega_n\subset \mathbb{R}^d$, where $\Omega_n$ is bounded, there is $N\subset\mathbb{N}$ infinite such that $\phi({y}_i)=\lim_{n\in N}\phi_n({y}_i)$ exists for all ${y}_i\in\Omega=\cup_{n=1}^{\infty}\Omega_n$.

Passing to the limit along $n\in N$ where $N$ is infinite and
$N\subset\mathbb{N}$, we obtain $\|\phi({y}_i)-\phi({y}_j)\|\leq
\|\phi({y}_i)-\phi({y}_k)\|$. Hence, $\phi$ is weakly isotonic on
$\Omega$ and by Theorem 1 in \cite{arias2017some}, there exists a
similarity transformation that coincides with $\phi$ on $\Omega$. That
is, there exists constant $s>0$, orthogonal matrix
$W\in\mathbb{R}^{d\times d}$ and  constant vector
${t}\in\mathbb{R}^{d}$ such that as $n\to\infty$, for all pairs
$(i,j)$ with $1\leq i<j\leq n$,  $$\|{x}_i-sW{y}_i-{t}\|\to 0.$$
We therefore have
\begin{align*}
   0\leq n^{-1/2}\|X-sYW-1t^{\top}\|_F&=\Bigl(n^{-1}\sum_{i=1}^n\|{x}_i-sW{y}_i-{t}\|^2\Bigr)^{1/2}\\
   &\leq \max_i\|{x}_i-s{y}_iW-{t}\|\to 0 
\end{align*}
as desired.
\end{proof}  
\subsection{Additional Simulation Study}
\textbf{\textit{Simulation 4: Sparsity.}} This
simulation is designed to investigate the the influence of sparsity of networks. Set dimension of embedding $K=3$, sparsity level $\rho$ satisfying $\rho= (\gamma \log{n})/n$ where $\gamma \in\{3,5\}$, $n\in\{100,200,500,1000\}$ and significant level $\alpha =0.05$. 
The latent positions are set to be $Y=X+Z$ where $Z=(z_{ij})\in\mathbb{R}^{n\times 2}$ and $z_{ij}\overset{iid}{\sim}  N(0,\epsilon)$ and $\epsilon\in\{0,0.02,0.1,0.2,0.5,1\}$. 

\begin{table}[H]
\centering
\def~{\hphantom{0}}
\caption{Power of the proposed test under different sparsity levels ($\alpha=0.05$).}{
\begin{tabular}{cccccccc}
  $\rho$&$n$&$\epsilon=0$&$\epsilon=0.02$&$\epsilon=0.1$&$\epsilon=0.2$&$\epsilon=0.5$&$\epsilon=1$\\
  \multirow{4}{*}{$\frac{3 \log{n}}{n}$}
  &100&0.05&0.15&0.29&0.37&0.88&0.97\\
  &200&0.05&0.15&0.01&0.64&0.96&1\\
  &500&0.05&0.29&0.50&0.92&1&1\\
  &1000&0.05&0.06&0.46&0.94&1&1\\
  \multirow{4}{*}{$\frac{5 \log{n}}{n}$}
  &100&0.05&0.05&0.14&0.59&0.87&1\\
  &200&0.05&0.01&0.35&0.26&1&1\\
  &500&0.05&0.07&0.33&1&1&1\\
  &1000&0.05&0.24&1&1&1&1\\
   \end{tabular}}
\label{sparstable}
\end{table}

Table \ref{sparstable} illustrates the performance of our test under different sparsity levels. Overall, it performs quite well especially in the mild sparse case with $\rho\geq 5$. When the network becomes more sparse ($\rho = 3$), the proposed test is not stable in the case with minor difference ($\epsilon \leq 0.1$) while it becomes much better and more robust as the difference is relatively larger ($\epsilon\geq0.5$) , for example, it has power 1 when $n$ is 500.

\end{document}